\documentclass[journal]{IEEEtran}

\usepackage[utf8]{inputenc}

\usepackage{amsmath, amssymb}
\usepackage{amsfonts}

\usepackage{makecell}
\usepackage{multirow}
\usepackage{threeparttable}
\usepackage{balance}
\usepackage{flushend}
\usepackage{changepage}
\usepackage{tikz}
\usepackage{diagbox}
\usepackage{array}
\usepackage{mathrsfs}
\usepackage{subfigure}
\usepackage{algorithm}  
\usepackage{algorithmicx}  
\usepackage{algpseudocode}  
\usepackage{multirow} 
\usepackage{xcolor}

\usepackage{graphicx}
\usepackage{color}
\usepackage{url}
\usepackage{bm}
\usepackage[all]{xy}
\usepackage{arydshln}
\usepackage{float}

\usepackage{enumerate}

\newtheorem{theorem}{Theorem}
\newtheorem{lemma}{Lemma}
\newtheorem{assumption}{Assumption}
\newtheorem{proof}{Proof}

\floatname{algorithm}{Algorithm}

\ifCLASSINFOpdf
\else
\fi
\begin{document}
%
\title{PPT: A Privacy-Preserving Global Model Training Protocol for Federated Learning in P2P Networks}
%
%
%

\author{Qian~Chen,
	Zilong~Wang,~\IEEEmembership{Member,~IEEE},
	Wenjing~Zhang,
	and~Xiaodong~Lin,~\IEEEmembership{Fellow,~IEEE}
	
	\thanks{Qian~Chen and Zilong~Wang are with the State Key Laboratory of Integrated Service Networks, School of Cyber Engineering, Xidian University, Xi'an,
		China (e-mail: xidianqianchen@gmail.com; zlwang@xidian.edu.cn).}
	\thanks{Wenjing~Zhang and Xiaodong~Lin are with the School of Computer Science, University of Guelph, Guelph, Canada (e-maill: wzhang25@uoguelph.ca; xlin08@uoguelph.ca).}
}
\maketitle

\begin{abstract}
The concept of Federated Learning (FL) has emerged as a convergence of machine learning,  information, and communication technology. It is vital to the development of machine learning, which is expected to be fully decentralized,  privacy-preserving, secure, and robust. However, general federated learning settings with a central server can't meet requirements in decentralized environment. In this paper, we propose a decentralized, secure and privacy-preserving global model training protocol, named PPT, for federated learning in Peer-to-peer (P2P) Networks. PPT uses a one-hop communication form to aggregate local model update parameters and adopts the symmetric cryptosystem to ensure security. It is worth mentioning that PPT modifies the Eschenauer-Gligor (E-G) scheme to distribute keys for encryption. In terms of privacy preservation, PPT generates random noise to disturb local model update parameters. The noise is eliminated ultimately, which ensures the global model performance compared with other noise-based privacy-preserving methods in FL, e.g., differential privacy. PPT also adopts Game Theory to resist collusion attacks. Through extensive analysis, we demonstrate that PPT various security threats and preserve user privacy. Ingenious experiments demonstrate the utility and efficiency as well.
\end{abstract}


%
\IEEEpeerreviewmaketitle

\section{Introduction}
%
%
%
%
\IEEEPARstart{M}{odern} machine learning (ML) is achieving unprecedented performance in natural language processing \cite{NLP15}, computer vision \cite{CV16}, data mining \cite{DM14}, etc. However, along with the growing social privacy awareness and the rapid growth of data, ML seems to be hard to break new ground. Especially, the General Data Protection Regulation (GDPR) \cite{GDPR17} enforces strict limitations on handling users' private data. Both industry and academia began to find a way performing ML with the demand of privacy. Recently, the concept of federated learning (FL) \cite{FL17} has emerged and been recognized as the state-of-the-art distributed ML system, which provides privacy, security, communication efficiency, and improved robustness. FL allows users to collectively reap the benefits of shared models trained from rich data without the need to centrally store it. Massive data containing sensitive information, such as religion, income and e-mail, never leaves the suers' devices.

A general FL system consists of two parties: a central server and a group of clients. The central server includes a coordinator and an aggregator. The aggregator aggregates the local training results and updates the global model under the control of the coordinator, as shown in Fig. 1. At the beginning of an FL model training round, the central server distributes a pre-trained model. Subsequently, clients train the model using their local data and upload their model update parameters to the aggregator. Through aggregating the model update parameters, the central server updates the global model. The way of aggregating model update parameters reduces the communication overhead greatly, compared with collecting training data. To defend against powerful attackers obtaining sensitive information by executing model inversion attacks \cite{Inversion15}, FL system adopts some privacy-preserving techniques when uploading the model update parameters. Homomorphic Encryption (HE) \cite{Rivest78, Gentry09}, Differential Privacy (DP) \cite{Dwork06}, and Secure Multi-Party Computation (MPC) \cite{Yao82} are the common privacy-preserving techniques in FL.

\begin{figure}[htbp]
	\centering
	\vspace*{-10pt}
	\hspace*{-5pt}
	\includegraphics[width=10cm,height=6cm,trim=50 80 50 50,clip]{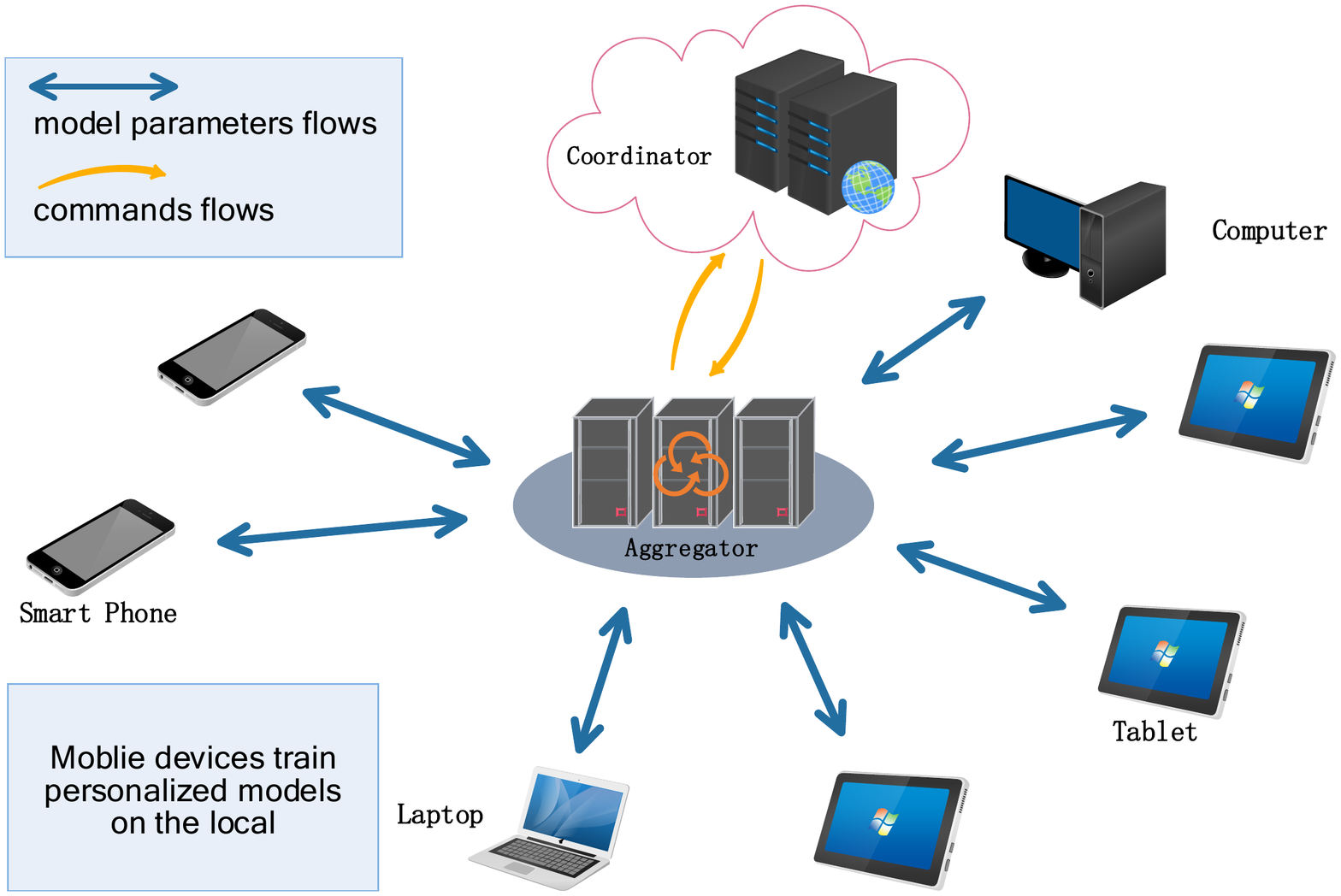}
	\caption{Federated Learning with Central Server.}     \label{fig-1}
\end{figure}

With the advantages of privacy preservation and communication efficiency, today's FL seems to be practical. However, in practice, clients are usually distributed in Peer-to-peer networks \cite{P2P01}, such as Smart Home, Internet of Things, and Ad. Hoc. There are large amounts of clients not connected with a central server directly, due to the geographical location, signal strength, and other possible reasons. Although this part of clients can communicate with each other, they can't upload the model update parameters to a central aggregator directly. For example, the devices in Smart Home are usually not connected to a smartphone, which plays the role of an aggregator, in day time. It means a lot of time that could be used for aggregation will be wasted.

Fortunately, a central server could not be essential in FL. As demonstrated by Lian {\em et al}. \cite{Lian17}, the central server may even become a bottleneck when the number of clients is very large. The survey in \cite{OpenProblem19} also proposes a decentralized environment could motivate the design of the next generation of FL systems. Subsequently, a series of FL frameworks and algorithms without a central aggregator was proposed. Roy {\em et al.} \cite{Braintorrent19} present BrainTorrent, a new FL framework in the highly dynamic P2P environment. Ramanan {\em et al.} \cite{BAFFLE20} propose a blockchain based aggregator free FL framework, BAFFLE, which achieves high scalability and computational efficiency in a private Ethereum network. Lu {\em et al.} \cite{Lu19} propose a fully decentralized FL framework by leveraging two classic non-convex decentralized optimizations. Lalitha {\em et al.} \cite{P2PFL19} considers the problem of training a machine learning model in a fully decentralized framework. The proposed algorithm generalizes the prior works on FL and obtains a theoretical guarantee (upper bounds) that the probability of error and true risk are both small for every participant. Dubey {\em et al.} \cite{Dubey20} devise FEDUCB for decentralized (peer-to-peer) FL to solve the contextual linear bandit problem. A fully decentralized FL approach proposed in \cite{Wittkopp21} introduces a concept of student and teacher roles for model training.

However, all the above advanced decentralized works only present frameworks or algorithms. A formal description of decentralized FL settings and a general model training protocol are sill absent. In this paper, the federated learning without a fully connected central server we refer to as {\it federated learning (FL) in peer-to-peer (P2P) networks}. We are interested in how to train a global model in the context of {\it FL in P2P networks} under the premise of privacy-preservation and security. Therefore, we design a \underline{P}rivacy-\underline{P}reserving global model \underline{T}raining (PPT) protocol for {\it FL in P2P networks}. PPT generates a random noise for local model disturbance to guarantee the privacy. And the noise will be eliminated ultimately, which ensure the global model performance. PPT also distributes communication keys for encrypted communication to enhance the security. Besides, PPT adopts Game Theory method to resist collusion attacks. And a broadcasting method is used throughout for the robustness of the system. The main contributions of this paper are as follows: 

\textbf{Federated learning in P2P networks.} We design a privacy-preserving global model training protocol, PPT, in the context of {\it FL in P2P networks}. To the best of our
knowledge, this is the first collaboratively training protocol in the context of {\em FL in P2P networks}. Following our PPT protocol, a group of clients connected in P2P networks can collaboratively train a global machine learning model privately and securely. And the global model performance is surely better than other FL aggregation methods using noise-disturbance-based privacy-preserving techniques, e.g., DP. 

textbf{Security and privacy-preserving protocol.} We analyze the security strength and privacy-preservation ability of PPT. In particular, we analyze the connectivity of the communication key establishment scheme, which is indispensable for the security of data transmission. 

\textbf{Experimental evaluation.} We evaluate the performance of the proposed PPT protocol by training a spam classification model on two real-world datasets, i.e. Trec06p and Trec07. We also research the client dropout situation to prove the robustness against the dropout problem. Experimental results show that the proposed protocol is more efficiency than Google's Secure Aggregation[]. 

The rest of the paper is organized as follows. Section II demonstrates some primitive concepts. The system model, security requirements, and the design goal are formalized in Section III. In Section IV, we design the PPT protocol. In Section V, we analyze the security and privacy-preservation of PPT, followed by the experimental design and the results and performance evaluation in Section VI and VII, respectively. The limitation and future work are shown in Section VIII. Finally, we draw our conclusions in Section IX.

\section{Preliminaries}

In this section, we briefly describe the data transmission form in P2P networks, the model training process, and the key pre-distribution scheme in the proposed PPT protocol. 

\subsection{Peer-to-peer networks and data transmission}

Peer-to-peer (P2P) networks were popularized by file sharing systems such as the music-sharing application Napster. A P2P network is a distributed application architecture that tasks or workloads are partitioned between clients. Clients are equally privileged, equipotent participants. Different from Client-server (CS) networks, P2P networks, in which interconnected clients share resources amongst each other without the usage of a centralized administrative system, don't need central coordination by servers or stable hosts. 

Data transmission in P2P networks usually follows a P2P transmission protocol, i.e., a sender upload the data, then any client connected to the sender can download the data directly. The transmission form is obviously efficient compared with CS networks.

\subsection{Stochastic gradient descent and model training}

Stochastic gradient descent (SGD) \cite{SGD10} is an iterative method for optimizing an objective function. In modern machine learning, SGD is usually regarded as a stochastic approximation of gradient descent optimization for loss function. Especially when the features are high-dimensional, SGD reduces the computational burden and achieves faster iterations to convergence.

In the FL context, $ n $ clients execute FederatedAveraging (or FedAvg) \cite{FL17} algorithm. That is, each client executes a fixed number of iterations of SGD on the current model using its local data, then the central server takes a weighted average of the resulting local models to update the global model. The global model update process is demonstrated as
\begin{equation}
	M^{R+1} \leftarrow \sum_{i=0}^{n-1}\frac{\omega_{i}}{\omega} {m_{i}^{R}} ,
\end{equation}  
where $ M^{R+1} $ is the updated global model in the $ R $th round, and $ m_{i}^{R} $ is the local model. The weights of the local model when aggregating are composed of the local training set size $ \omega_{i} $ and the sum of local training set sizes $ \omega=\sum_{i=0}^{n-1}\omega_{i} $.

\subsection{Eschenauer-Gligor scheme }

Eschenauer-Gligor(E-G) scheme \cite{EG02} is first proposed as a random key pre-distribution scheme for distributed sensor networks. The basic idea is that clients randomly pick a certain number of keys from a large key pool and use the same keys as the communication keys with other clients. As a key-management scheme, the E-G scheme requires memory storage for only few tens to a couple of hundred keys and provides similar security and superior operational properties comparing to pair-wise private key-sharing schemes.

A typical E-G scheme consists of the following three phases:

\begin{adjustwidth}{0.5cm}{0cm}
	i) \textit{Phase 1, key pre-distribution}: A trusted central server generates a large key pool containing $ \eta $ keys and key identifiers (IDs) offline. Each client randomly extracts $ l $ keys from the key pool with replacement to establish its own key ring and stores the key ring locally. 
\end{adjustwidth}

\begin{adjustwidth}{0.5cm}{0cm}    
	ii) \textit{Phase 2, shared-key discovery}: Each client broadcasts the lists of identifiers in the key ring to discover the same keys, called shared-keys, with its neighbor clients. 
\end{adjustwidth}

\begin{adjustwidth}{0.5cm}{0cm}     
	iii) \textit{Phase 3, path-key establishment}: A third-party client assigns a path-key to the pairs of clients that are connected but sharing no keys after the shared-key discovery phase.
	
\end{adjustwidth}

\section{System model, security requirement and design goal}

In this section, we formalize the system model, security requirements, and identify our design goal.

\subsection{System model}

In the context of {\it FL in P2P networks}, all clients $ c_{i} $ distributed in P2P networks are potential clients. We mark the clients willing to participate global model training as target clients. And only a small part (even if one client) of them are connected to an aggregator directly. In the system model, we mainly focus on how to collaboratively train a global machine learning model among target clients. Every target client $ u_{i} $ trains a local model $ m_{i} $ using a central original model $ M $ and its own data, where $ i=0,1,...,n-1 $. The local training processes can be seen as executing stochastic approximation of gradient descent, i.e., SGD, for the central original model. A local update to the original model by a target client can be considered a product of a gradient and a step size. In this paper, we simply define the local model update parameters as $ x_{i} $ and denote as:
\begin{equation}
x_{i}=m_{i}-M
\end{equation}

Considering the impact of different local models on the global model, different weights $ \omega_{i} $ are set for corresponding local model updating parameters. Thus the global model update process is demonstrated as:

\begin{eqnarray}\label{eq-3}
M^{new} = \frac{\sum_{i=0}^{n-1} \omega_{i}x_{i}}{\sum_{i=0}^{n-1} \omega_{i}} + M,
\end{eqnarray}
where $ M^{new} $ is the updated global model. We illustrate the system model in the left part of Fig. 2.

\begin{figure*}[htbp]
	\centering
	\vspace*{-30pt}
	\hspace*{0pt}
	\includegraphics[width=18cm,height=10cm,trim=50 100 50 50,clip]{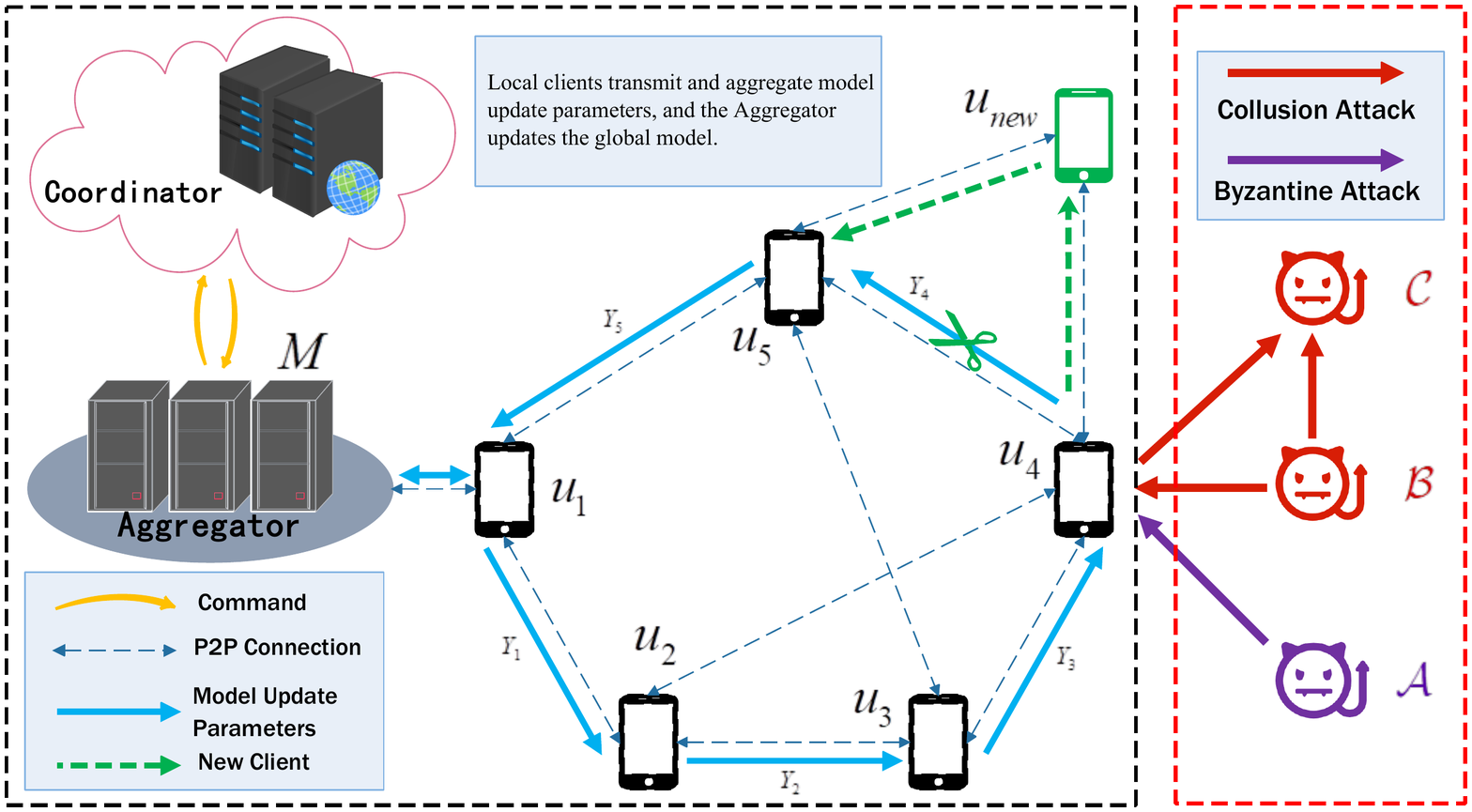}
	\caption{System model under consideration.}     \label{fig-2}
\end{figure*}

\textit{Communication model.} The process of aggregating data can be easily finished in CS networks, as all clients just need to upload $x_{i}$ to a central aggregator. When coming to P2P networks, it seems hard to transmit data to an aggregator, especially for clients not connected to the aggregator directly. However, a client can transmit the local model update parameters to a one-hop neighbor (neighbor client) directly within its communication range. The recipient adds its own data and transmits the aggregated data to its one-hop neighbor sequentially. We define this kind of transmission form as {\em one-hop transmission}.

\subsection{Security requirements}
Security and privacy preservation are crucial for the success of {\it FL in P2P networks}. In our security model, we consider the aggregator and clients are honest but curious. That is, local model update parameters are considered private. Meanwhile, there exists an adversary $\mathcal{A}$ aiming to broke the global model training process. In addition, the adversary $\mathcal{A}$ could also launch active attacks to threaten the data integrity. More seriously, two adversaries $\mathcal{B}$ and $\mathcal{C}$ could execute collusion attacks \cite{Collusion15} to eavesdrop on the model update parameters. For instance, the upstream and downstream clients of a client $ u_{i} $ are adversaries $\mathcal{B}$ and $\mathcal{C}$. $ \mathcal{B} $ transmits its data $ x_{\mathcal{B}} $ to $ u_{\mathcal{C}} $ directly after $ u_{\mathcal{C}} $ receives the data $ x_{\mathcal{B}}+x_{i} $ from $ u_{i} $. Whereupon, $ u_{\mathcal{C}} $ can calculate the data $ x_{i} $. We illustrate the adversary model in the right part of Fig. 2. To prevent adversaries from learning the private model update parameters and to defend against malicious actions, the following security requirements should be satisfied.

\begin{adjustwidth}{0.5cm}{0cm}
	\textit{Confidentiality.} Protect individual local model update parameters from anyone but the client itself. Even the aggregator can only read the aggregated results rather than individual local model update parameters. Moreover, the parameter privacy will not be compromised even if there is a collusion attack.
	
\end{adjustwidth}

\begin{adjustwidth}{0.5cm}{0cm}
	\textit{Authentication and Data integrity.} Authenticate the model update parameters that are really sent by a legal client and have not been altered during the transmission, i.e., if the adversary $\mathcal{A}$ forges and/or modifies the model update parameters, the malicious operations should be detected.
	
\end{adjustwidth}

\begin{adjustwidth}{0.5cm}{0cm}
	\textit{Byzantine robustness.} Defend against Byzantine attacks caused by consensus problems, i.e., if the adversary $\mathcal{A}$ claims having the aggregated model update parameters, the malicious operations should be detected.
	
\end{adjustwidth}

\subsection{Design goal}
Under the aforementioned system model and security requirements, our design goal is to design an efficient and privacy-preserving global model training protocol for {\it FL in P2P networks}. Specifically, the following two objectives should be achieved.

\begin{adjustwidth}{0.5cm}{0cm}
	\textit{The security requirements should be guaranteed in the proposed protocol.} As stated above, if the training process does not consider security, the users’ privacy could be disclosed, and the global model could be destroyed. Therefore, the proposed protocol should achieve confidentiality, authentication, data integrity, and Byzantine robustness \cite{Byzantine82} simultaneously.
	
\end{adjustwidth}

\begin{adjustwidth}{0.5cm}{0cm}
	\textit{Communication efficiency should be guaranteed in the proposed protocol.} Although the communication among users in P2P networks is featured with high efficiency, to support hundreds and thousands of target clients aggregating local model update parameters, the arranged aggregating route should also consider the communication efficiency.
	
\end{adjustwidth}

\section{System design}

In this section, we propose an efficient and privacy-preserving global model training protocol for {\it FL in P2P networks}, which mainly consists of the following five parts: communication key establishment, local model training and disturbance, model update parameters transmission and aggregation, global model update, and complementary mechanisms.

\subsection{Communication key establishment} 

We first construct communication keys for potential clients. The operations that broadcasting the IDs of the keys in the key ring and path-key distribution leak the privacy of shared-keys obviously. We modify the E-G scheme and used it as our communication key establishment scheme. 
A large key pool is generated previously, which contains $ \eta $ keys and their identifiers (IDs). Every potential client $ c_{i} $ randomly extracts $ l $ keys from the key pool with replacement to establish the key ring $ \{k_{i\alpha}\vert\alpha=0,1,...,l-1\} $. Afterwards, each client $ c_{i} $ broadcasts encrypted messages $ \{A_{i\alpha}\vert \alpha=0,1,...,l-1\} $. Every $ A_{i\alpha} $ is encrypted by every $ k_{i\alpha} $ in $ c_{i} $'s own key ring, and the encryption scheme is the symmetric cryptosystem. We represent the encryption process as 
\begin{equation*}
	A_{i\alpha}\leftarrow \textbf{\emph{EG.Enc}}\left(a_{i\alpha}, k_{i\alpha}\right),
\end{equation*} 
where $ a_{i\alpha} $ is the plaintext of the message. The client $ c_{j} $ who receives $ A_{i\alpha} $ reveals the challenge by executing the decryption process, shown as 
\begin{equation*}
	a_{i\alpha} \leftarrow \textbf{\emph{EG.Dec}}\left(A_{i\alpha}, k_{j\alpha}\right).
\end{equation*}

Thereby, all neighbor clients can get the knowledge of keys held by both the broadcast client and themselves, respectively. All clients execute the above interactions mutually to obtain the same keys with each other, which are called shared-keys $ \{k_{i,j}^{r}\vert r\in \mathbb{N^{*}}\} $. Any two adjacent clients $ c_{i} $ and $ c_{j} $ with more than $ e $ shared-keys execute XOR operations to compute their communication key, demonstrated as:

\begin{equation*}
K_{i,j}=k_{i,j}^{1} \oplus k_{i,j}^{2} \oplus ... \oplus k_{i,j}^{e}\oplus...  \quad. 
\end{equation*}

As for two clients connected directly but having no shared-keys, a third-party client distributes the keys that haven't been used to the two clients. Then, they choose some keys independently to obtain shared-keys and generate communication keys for themselves. A client will never choose all the received keys from the third-party client as its shared-keys. Otherwise, it will lead to privacy leakage, as the third-party client will calculate their communication key. Therefore, the communication keys are established for all potential clients.

In addition, we propose a key {\em revocation and update mechanism} to enhance the communication key establishment scheme. When a client is recognized as a malicious client, we execute the following operations to revoke the poison keys hold by the malicious client and rebuild communication keys for other participants:
\begin{adjustwidth}{0.5cm}{0cm}	
	\noindent \textcircled{\oldstylenums{1}} Coordinator broadcasts a countermand which contains the IDs of the poison keys.
	
	\noindent \textcircled{\oldstylenums{2}} Clients delete the invalid keys from their key rings.
	
	\noindent \textcircled{\oldstylenums{3}} Affected clients rebuild communication keys.
	
	\noindent \textcircled{\oldstylenums{4}} When the key pool is used for a long time, the lifetime of keys expires. All clients restart the communication key establishment phase. 
\end{adjustwidth}

\subsection{Local model training and disturbance} 

For the settings of {\it FL in P2P networks}, it is reasonable to assume an honest but curious server can coordinate the whole system. When the coordinator starts a training process for a global model, all target clients download the central original model $ M $ following the inherent P2P transmission protocol.

After local training, target clients obtain local personalized models $ m_{i} $ and prepare to uploads the weighted local model update parameters $\omega_{i}x_{i}$ and weights $\omega_{i}$. To improve communication efficiency, clients encode the uploading data by attaching $\omega_{i}$ to the end of $ \omega_{i}x_{i} $, which is defined by $X_{i}$, and denoted by:
\begin{equation}
X_{i}=(\omega_{i}x_{i},\omega_{i}) . 
\end{equation}

Then, the server chooses a target client which is connected to the server directly as the leader client $ u_{0} $. The transmission and aggregation route will start from the leader client. To protect $u_{0}$'s privacy, $u_{0}$ generates a noise $ s $ to disturb $ X_{0} $, shown as $ X_{0}+s $. Note that $ s $ has the same dimension as $ X_{0} $. Subsequently, the leader client $ u_{0} $ can execute the transmission and aggregation process.

\subsection{Model update parameters transmission and aggregation} 

The whole data transmission is following the inherent P2P transmission protocol, i.e., the sender uploads the data, then the recipient downloads the data directly. It is not conducive to planning a transmission route, as all neighbor clients can download the data from the sender.

Therefore, the leader client $ u_{0} $ chooses a neighbor client as the downstream client $ u_{1} $. $ u_{0} $ encrypts $ X_{0}+s $ using their communication key $K_{0,1}$. The encryption scheme is the symmetric cryptosystem. We write the encryption process as: 
\begin{equation*}
Y_{0}\leftarrow \textbf{\emph{Enc}}(X_{0}+s,K_{0,1}) . 
\end{equation*} 
Then, $u_{0}$ uploads $Y_{0}$. Although every $ u_{0} $'s neighbor client can still download $ Y_{0} $, only the chosen client $u_{1}$ can decrypt $Y_{0}$. 

After downloading $ Y_{0} $, $ u_{1} $ executes the decryption process, shown as: 
\begin{equation*}
X_{0}+s\leftarrow \textbf{\emph{Dec}}(Y_{0},K_{0,1}) . 
\end{equation*} 
Afterward, $ u_{1} $ adds its own data $ X_{1}=(\omega_{1}x_{1},\omega_{1})$, denoted by $ X_{0}+s+X_{1} $. Whereafter, $ u_{1} $ chooses a new neighbor client and executes encryption and uploading operations similar to $ u_{0} $.

The rest of target clients will execute the above operations successively until all participants finish data aggregation. Then, the aggregating data will be transmitted back to $u_{0}$. $ u_{0} $ subtracts the noise $ s $ and obtains the final aggregated data 
\begin{equation}
\sum_{i=0}^{n-1}X_{i}=X_{0}+X_{1}+...+X_{n-1} . 
\end{equation}

The choice of the downstream client is the most important factor affecting aggregation efficiency. PPT adopts the {\em depth-first search algorithm in the graph} to enhance efficiency. The upstream client will always choose an unvisited client as the downstream client if possible. The algorithm explores as far as possible before backtracking.

In addition, due to the existence of an active adversary $\mathcal{A}$ threatening the data integrity, a reliable digital signature is essential. In the proposed PPT protocol, we improve the function of digital signatures by signing a time stamp together with data to ensure timeliness. Noted that the digital signature scheme is the asymmetric cryptosystem, and the secret key $ SK_{i} $ and public key $ PK_{i} $ are initially deployed on every client $ u_{i} $. The details of our signature scheme are as follows:
\begin{adjustwidth}{0.5cm}{0cm}
\noindent \textcircled{\oldstylenums{1}} The data holder $ u_{i} $ attaches a timestamp $ t_{i} $ to the encrypted data $ Y_{i} $ and signs on it using its secret key $ SK_{i} $, which is shown as 
\begin{equation*}
\sigma_{i} \leftarrow  \textbf{\emph{Sign}}(Y_{i},t_{i},SK_{i})  .
\end{equation*}

\noindent \textcircled{\oldstylenums{2}} $ u_{i} $ sends the encrypted data with signature  $ <\!Y_{i},\sigma_{i}\!> $ to the downstream client.

\noindent \textcircled{\oldstylenums{3}} The recipient verifies the signature by $ u_{i} $'s public key $ PK_{i} $ and check the time stamp, which is demonstrated as
\begin{equation*}
\{1,0\}\leftarrow \textbf{\emph{Verf}}(\sigma_{i},PK_{i}) .
\end{equation*}
\end{adjustwidth}

\subsection{Global model update} 

After local model update parameters aggregation, the leader client $u_{0}$ uploads the aggregated data $ \sum_{i=0}^{n-1}X_{i} $. The uploading data is encrypted and the signature is surely attached. The server confirms the aggregated data by verifying the signature and decryption. Then, the server decodes the aggregated data to aggregated weighted model update parameters $ \sum_{i=0}^{n-1}\omega_{i}x_{i} $ and aggregated weights $ \sum_{i=0}^{n-1}\omega_{i} $. Whereupon, the server updates the global model. 
The form of global model update is shown in Equation (\ref{eq-3}). Afterward, the updated global model will be distributed to all target clients to train and aggregate over and over again until the global model is converged.
We demonstrate the global model update process in the $ R $ round as follows:

\begin{equation}
M^{R+1}=\frac{\sum_{i=0}^{n-1} \omega_{i}x_{i}^{R}}{\sum_{i=0}^{n-1} \omega_{i}}+M^{R}
\end{equation}

\subsection{Complementary mechanisms}

The above design does well in efficiency and security when all participants are honest but curious. But in practice, there still exists active adversaries threatening the aggregation and transmission process. Besides, unlimited transmission and new client participation are not allowed. Hence, several complementary mechanisms are annexed to our basic design. The details of the complementary mechanisms are described as follows:

\emph{Neighborhood Broadcast Mechanism}: To enhance the robustness and security, we propose a Neighborhood Broadcast mechanism. The mechanism dictates that every target client broadcasts the behavior and a timestamp $ \tau $ to all neighbor clients when executing an operation. The neighbor clients sequentially pass the broadcast to their neighbors until all clients receive the message. Based on the neighborhood broadcast mechanism, all operations are under the supervision of all clients, including the dropout and the new participation.

\emph{Termination Mechanism}: In practice, the time of one global model training round is fixed and limited. Unlimited transmission, aggregation, and new client participation are not allowed. Thus, the following termination operations will be executed:

\begin{adjustwidth}{0.5cm}{0cm}	
	\noindent \textcircled{\oldstylenums{1}} The request of a newly joined client will be refused when one-half of the target clients have completed the aggregation process.
	
	\noindent \textcircled{\oldstylenums{2}} PPT dictates executing the backtracking operation immediately when the time has passed two-thirds of the stipulated aggregation time.
	
\end{adjustwidth}

\emph{Supervision and Report Mechanism}: The collusion attack is common in distributed systems, which is no exception in the context of {\em FL in P2P networks}. To prevent collusion attacks, we propose the Supervision and Report Mechanism based on Game Theory \cite{Nash50}. The core idea is encouraging mutual reporting, and the process is as follows:
	
\begin{adjustwidth}{0.5cm}{0cm}	
	\noindent \textcircled{\oldstylenums{1}} All participants are required to pay a deposit $ d $ before starting the aggregation process. And the deposit will be returned at the end of the aggregation.
	
	\noindent \textcircled{\oldstylenums{2}} We encourage mutual reporting aiming at malicious operations. If the alleged malicious operation is proven to be true, the defendant's deposit will be paid to the complainant.
	
	\noindent \textcircled{\oldstylenums{3}} If two participants report each other to get a pay, their deposits of them will be confiscated.

\end{adjustwidth}

	\noindent As long as the profit by collusion attack is lower than the deposit, i.e., $ g\textless d $, the two parties will tend to execute the PPT protocol honestly. We analyze security based on Nash Equilibrium in Section V.

Integrating the main processes mentioned above and the complementary mechanisms, we propose our \underline{P}rivacy-\underline{P}reserving Global Model \underline{T}raining (PPT) protocol in Fig. 3. Note that the red underlined parts are required to guarantee security under the active-adversary assumption (and not necessarily under the honest-but-curious assumption).

\begin{figure*}\label{fig-3}
	\centering 
	\begin{tikzpicture}[scale=0.7]
	a client which is chosen to be the downstream client drops	
	\path[fill=yellow!0, draw=black!50](0,0) rectangle (25,33);
	
	\node[right] at (5.0,32.5){{Privacy-Preserving Global Model Training Protocol for FL in P2P networks}};
	
	\draw[thick](1,32.1)-- (24,32.1);
	
	\node[right] at (1,31.7){\textbf{\small{ $\bullet$\quad Setup:}}};
	
	\node[right] at (2,31.2){\footnotesize{- A large key pool containing $ \eta $ keys is generated offline.}};
	
	\node[right] at (2,30.7){\footnotesize{- Every potential participant $ c_{i} $ randomly extracts $ l $ keys from the key pool with replacement.}};

	\node[right] at (2,30.2){\footnotesize{- $ c_{i} $ stores the extracting keys and the IDs of the keys to establish the key ring $\{k_{i\alpha}\vert \alpha\!=\!0,\!1,\!...,\!l-1\} $.}};

	\node[right] at (2,29.7){\footnotesize{- $ \theta $ kinds of noise generation algorithms are built in every client in advance as well as the public \& private key $ P\!K\!_{i} $ and $S\!K\!_{i}$.}};
	
	\node[right] at (1,29.2){\textbf{\small{ $\bullet$\quad Phase 0 (Communication key establishment):}}};
	
	\node[right] at (2,28.7){\footnotesize{- $ c_{i} $  broadcasts $ l $ encrypted messages $ \{A_{i\alpha}\} $ using every $ k_{i\alpha} $ in $\{k_{i\alpha}\}$ respectively, shown as $ A_{i\alpha}\leftarrow \textbf{\textit{EG.Enc}}(a_{i},k_{i\alpha}) $.}};

	\node[right] at (2,28.2){\footnotesize{- The client $ c_{j} $ who receives $ A_{i\alpha} $s tries to decrypt the messages, shown as $ a_{i}\!\!\leftarrow\! \textbf{\textit{EG.Dec}}(A_{i\alpha},k_{j\alpha}) $. Thereby, all recipient clients  }};
	
	\node[right] at (2.3,27.7){\footnotesize{get the keys held by both the broadcast client and themselves.}};

	\node[right] at (2,27.2){\footnotesize{- All potential clients execute the above interactions mutually to obtain the shared-keys $ \{k_{ij}^{r}\vert r\in \mathbb{N^{*}}, r\textless e\} $.}};
	
	\node[right] at (2,26.7){\footnotesize{- Adjacent two clients $ c_{i} $ and $ c_{j} $ with more than $ e $ shared-keys execute XOR operations to compute their communication key, shown }};
	
	\node[right] at (2.3,26.2){\footnotesize{as $ K_{i,j}=k_{i,j}^{1} \oplus k_{i,j}^{2} \oplus ... \oplus k_{i,j}^{e}\oplus...  $ }};

	\node[right] at (2,25.7){\footnotesize{- Consider the situation that two clients are connected directly but don't have shared-keys, a third party client distributes the keys }};
	
	\node[right] at (2.3,25.2){\footnotesize{that haven't been used for the shared-key establishment to the two clients. Then, they choose some keys independently to obtain}};
	
	\node[right] at (2.3,24.7){\footnotesize{shared-keys and generate the communication key.}};

	\node[right] at (2,24.2){\footnotesize{- {\color{red}\underline{If a client is recognized as a malicious client, the affected clients execute the Key Revocation and Update Mechanism to rebuild}} }};
	
	\node[right] at (2.3,23.7){\footnotesize{{\color{red}\underline{communication keys.}}}};
	
	\node[right] at (2,23.2){\footnotesize{- When the key pool is used for a long time, the lifetime of keys expires. All potential participants execute the Key Revocation and  }};

	\node[right] at (2.3,22.7){\footnotesize{Update Mechanism to rebuild communication keys }};

	\node[right] at (2,22.2){\footnotesize{- If more than half of the potential clients are identified as malicious clients, abort.}};
	
	\node[right] at (1,21.7){\textbf{\small{ $\bullet$\quad Phase 1 (Local model training and disturbance):}}};
	
	\node[right] at (2,21.2){\footnotesize{- Participants broadcast that they complete local training and will participate in the global model update. }};
	
	\node[right] at (2,20.7){\footnotesize{- The coordinator chooses a client as the leader client $u_{0}$ randomly. (To increase the success rate of the protocol, the coordinator}};
	
	\node[right] at (2.3,20.2){\footnotesize{prefers a client that has completed the protocol before as the leader client.)}};
	
	\node[right] at (2,19.7){\footnotesize{- The leader client $ u_{0} $ generates a disturbed noise $ s $ by the built-in noise generation algorithms and disturbs the encoding local model }};
	
	\node[right] at (2.3,19.2){\footnotesize{updating parameters $ X_{0} $, which is shown as $X_{1}+s$.}};

	\node[right] at (2,18.7){\footnotesize{- If the leader client drops out, abort. The coordinator selects a new leader client and restarts Phase 1. }};

	\node[right] at (1,18.2){\textbf{\small{ $\bullet$\quad Phase 2 (Model update parameters transmission and aggregation):}}};
	
	\node[right] at (2,17.7){\footnotesize{- $ u_{0} $ chooses a neighbor $ u_{1} $ and broadcasts a message to all neighbors that it will transmit the aggregated data to $ u_{1} $.  }};
	
	\node[right] at (2.3,17.2){\footnotesize{A timestamp $ \tau_{0} $ is attached to the message.}};

	\node[right] at (2,16.7){\footnotesize{- $ u_{0} $ uses its corresponding communication key $  K_{0,1}$ to encrypt the disturbed data, shown as $Y_{0}\leftarrow \textbf{\textit{Enc}}(X_{0}+s,K_{0,1}) $.}};
	
	\node[right] at (2,16.2){\footnotesize{- {\color{red}\underline{$ u_{0} $ attaches a timestamp $ t_{0} $ to the ciphertext $ Y_{0} $ and signs on it, shown as $\sigma_{0}\leftarrow  \textbf{\textit{Sign}}(Y_{0},t_{0},SK_{0})  $}}.}};
	
	\node[right] at (2,15.7){\footnotesize{- $ u_{0} $ sends the encrypted data {\color{red}\underline{with the signature $ <\!\!Y_{0},\sigma_{0}\!\!> $}} to $ u_{1} $.}};
	
	\node[right] at (2,15.2){\footnotesize{- {\color{red}\underline{When $u_{1}$ receives the data, it verifies the signature and checks the timestamps, shown as $ \{1,0\}\!\leftarrow\! \textbf{\textit{Verf}}(\tau_{0},\sigma_{0},\!P\!K_{0}\!) $}}.}};

	\node[right] at (2,14.7){\footnotesize{- {\color{red}\underline{If $ u_{1} $ confirms that the data truly comes from $u_{0}$}}, $u_{1}$ decrypts the encrypted data, shown as $X_{0}+s\leftarrow \textbf{\textit{Dec}}(Y_{0},K_{0,1}) $.}};
	
	\node[right] at (2.3,14.2){\footnotesize{{\color{red}\underline{If not, wait for the correct aggregated data from $ u_{0} $}}.}};
	
	\node[right] at (2,13.7){\footnotesize{- $u_{1}$ aggregates the plaintext $X_{0}+s$ and its own data $ X_{1} $, shown as $ X_{0}\!\!+\!X_{1}\!+\!s $.}};
	
	\node[right] at (2,13.2){\footnotesize{- $ u_{1} $ executes the above broadcast, encryption, {\color{red}\underline{timestamp attachment, signing,}} and transmission operations as $ u_{0} $ did.}};
	
	\node[right] at (2,12.7){\footnotesize{- If $ u_{0} $ doesn't receive the broadcast from $ u_{1} $, $ u_{0} $ chooses a new neighbor as $u_{1}$ and executes the data aggregation process.  }};
	
	\node[right] at (2,12.2){\footnotesize{- The same operations are executed among all participants successively until all participants complete aggregation. }};

	\node[right] at (2,11.7){\footnotesize{- The selection of the downstream client follows the {\em depth-first search algorithm}.}};
	
	\node[right] at (2,11.2){\footnotesize{- If all neighbor clients have received the aggregating data already, the current client passes the data back to the previous client.}};
	
	\node[right] at (2,10.7){\footnotesize{- The protocol dictates that every honest user can only aggregate the local model update parameters once. }};
	
	\node[right] at (2.3,10.2){\footnotesize{As for the clients receiving encrypted data again, they only execute {\color{red}\underline{signature verification}}, decryption, encryption, {\color{red}\underline{signing}}, broadcast,}};
	
	\node[right] at (2.3,9.7){\footnotesize{and transmission operations.}};
	
	\node[right] at (2,9.2){\footnotesize{- Finally, $ u_{0}$ receives the aggregated result $\sum_{i=0}^{n-1}X_{i}+s$}};

	\node[right] at (2,8.7){\footnotesize{- $u_{0}$ subtracts the disturbed noise $s$ and obtains the encoding global model update parameters $X_{0}+X_{2}+...+X_{n-1}$.}};
	
	\node[right] at (2,8.2){\footnotesize{- If $ u_{0} $ doesn't receive aggregated result after the stipulated time, abort.}};
		
	\node[right] at (1,7.7){\textbf{\small{ $\bullet$\quad Phase 3 (Global model update):}}};

	\node[right] at (2,7.2){\footnotesize{- $u_{0}$ uploads the encoding global model update parameters $X_{0}+X_{2}+...+X_{n-1}$ to the aggregator.}};
	
	\node[right] at (2,6.7){\footnotesize{- The aggregator decodes the encoding global model updating parameters to $\sum_{i=0}^{n-1}\omega_{i}x_{i}$ and $\sum_{i=0}^{n-1}\omega_{i}$.}};
	
	\node[right] at (2,6.2){\footnotesize{- The aggregator updates the global model and broadcasts a message that the global model has been updated. }};
	
	\node[right] at (2,5.7){\footnotesize{- If the aggregator doesn't receive the model updating parameters after the stipulated time, abort.}};
	
	\node[right] at (1.2,5.2){\small{\textsl{All clients follow the {\color{red}\underline{Supervision and Report Mechanism}} in the whole data aggregation process.}}};
	
	\node[right] at (1.2,4.5){\textbf{\small{PS: }}};
	
	\node[right] at (1.0,4.0){\textbf{\small{ $\bullet$\quad For a newly joined client:}}};

	\node[right] at (2,3.5){\footnotesize{- Start the Communication key establishment phase for the newly joined client $ u_{n\!e\!w}$ }};
	
	\node[right] at (2,3.0){\footnotesize{- $ u_{n\!e\!w}$ requests for participating in the global model training process.}};
	
	\node[right] at (2,2.5){\footnotesize{- $ u_{n\!e\!w}$ joins in the data transmission process according to the {\em depth-first search algorithm}.}};
	
	\node[right] at (1.0,2.0){\textbf{\small{ $\bullet$\quad Termination:}}};
	
	\node[right] at (2,1.5){\footnotesize{- The protocol rejects the request of a newly joined client when one-half of the clients have completed the aggregation process.}};
	
	\node[right] at (2,1.0){\footnotesize{- The protocol executes backtracking operation immediately when the time has passed two-thirds of the stipulated aggregation time.}};

	\end{tikzpicture}
	\caption{Detailed description of the Privacy-preserving Training Protocol. {\color{red}\underline{Red, underlined parts are required to guarantee security under the active-adversary}} \protect\\{\color{red}\underline{assumption (and not necessarily under the honest-but-curious assumption).}} }
\end{figure*} 

\section{Security Analysis}
In this section, we analyze the security properties of the proposed PPT protocol. In particular, following the security requirements discussed earlier, our analysis will focus on how the proposed PPT protocol can achieve local model update parameters privacy preservation, source authentication and data integrity, and Byzantine robustness. Besides, we analyze the connectivity based on the communication key establishment scheme, which is indispensable for the security of data transmission.

\subsection{Privacy of clients}

\textit{The local model update parameters are privacy-preserving.} According to the proposed PPT protocol, the local model update parameters are uploaded in the form of ciphertext. Only the client chosen by the sender can decrypt the ciphertext using the communication key. For any other honest-but-curious client and server in P2P networks, as long as the encryption algorithm is secure, it is impossible to obtain any information of the model update parameters. As for a chosen downstream client $ u_{\nu} $, he can only obtain aggregated parameters with noise $ \sum_{i=0}^{\nu-1}{X_{i}+s} $.

\textit{The authentication and data integrity of clients' model update parameters are achieved in the proposed PPT protocol.} In the proposed PPT protocol, each client signs on the ciphertext, when uploading the aggregated parameters. Therefore, if the adopted digital signature scheme is provably secure, the source authentication and data integrity can be guaranteed. As s result, the adversary $\mathcal{A}$’s malicious behaviors can be detected.

\textit{The Byzantine robustness is achieved in the proposed PPT protocol.} In the proposed PPT
protocol, All operations are under the supervision of all clients based on the neighborhood broadcast mechanism. Only the operation achieving consensus by all clients can be accepted. Even if an adversary $\mathcal{A}$ forges aggregated parameters to confuse a target client, the target client can confirm the correct aggregated parameters by comparing received broadcasts from other clients.

\subsection{Defense against collusion attack}

\textit{The proposed PPT protocol is a strong defense against collusion attacks.} All clients abide by the Supervision and Report mechanism. Regarding the collusion attack by adversaries $\mathcal{B}$ and $\mathcal{C}$ as a game based on Game Theory. There is a basic assumption in Game Theory that the participants are selfish but rational. We generalize the assumption to the context of {\em FL in P2P networks}.  
\begin{assumption}\label{asp-1}
	The adversaries $\mathcal{B}$ and $\mathcal{C}$ are selfish but rational in collusion attacks.
\end{assumption}
\begin{theorem}\label{them-1}
	A client's model update parameters are private against collusion attacks under {\it Assumption \ref{asp-1}}.
\end{theorem}
\begin{proof}
	According to the Supervision and Report mechanism, $\mathcal{B}$ and $\mathcal{C}$ have the same two strategies, i.e., collusion and counter-collusion. We assume the deposits of the two parties $\mathcal{B}$ and $\mathcal{C}$ are $ d $, and the profits of executing collusion attacks successfully are $ g $. Thereby, we can get the pay-off matrix in TABLE I.
	\begin{table}[htbp]
		\centering
		\caption{Pay-off Matrix} 
		
		\begin{tabular}{|m{2cm}<{\centering}|m{1.6cm}<{\centering}m{0pt}|m{2cm}<{\centering}|}
			\hline
			\diagbox[width=8.5em]{$\mathcal{B}$}{$\mathcal{C}$} & \quad\quad collusion&\rule{0pt}{13pt} & counter-collusion  \\  
			\hline
			collusion & \quad\quad$ (g,g) $&\rule{0pt}{13pt} & $ (-d,d) $ \\
			\hline
			counter-collusion & \quad\quad$ (d,-d) $&\rule{0pt}{13pt} & {\color{red}\underline{$ (0,0) $}} \\
			\hline
		\end{tabular}
		
	\end{table}
	
	We definitely hope there is no collusion attacks, i.e., the Nash Equilibrium \cite{Nash50} of the game should be {\color{red}\underline{(counter-collusion, counter-collusion)}}. Therefore, as long as $ g\textless d $, the two adversaries tend to follow the proposed PPT protocol honestly.
	
\end{proof}

\subsection{Analysis of full connectivity }

The limits of the inherent P2P transmission protocol preclude the transmission of private data. Ensuring the secure connectivity rate based on communication keys being maximum possible is of vital importance. We first recall the monotonicity of Random Graph \cite{Graph60}.

Let $G(n, p)$ is a random graph, where $ n $ is the number of clients and $ p $ is the probability that a link exists between two clients.

\begin{lemma} \label{lem-1}
	Given a desired connectivity probability $ P_{c} $ for a graph $G(n, p)$, the threshold function $ p $ is defined by:
	\begin{eqnarray}
	P_{c}=\lim _{n \rightarrow \infty} \operatorname{Pr}[G(n, p) \text { is connected }]=e^{-e^{-c}},
	\end{eqnarray}
	Where $ p=\frac{\ln n }{n} +\frac{c}{n}  $ and $ c $ is any real constant.
\end{lemma}

Laurent {\it et al.} \cite{EG02} give the trade-off between the sizes of the key pool and the key ring.

\begin{lemma} \label{lem-2}
	For a given $ p $, the trade-off between the key pool size and the key ring size follows the equality:
	\begin{eqnarray}
	p = 1-\frac{\left((\eta-l) !\right)^{2}}{(\eta-2 l) ! \;\eta !},
	\end{eqnarray}
	Where $ \eta $ is the size of the key pool, and $ l $ is the size of the key ring.
\end{lemma}

\begin{theorem} \label{them-2}
	The communication key establishment phase in the proposed PPT protocol satisfies the requirement of private data transmission by choosing proper sizes of the key pool and the key ring.
\end{theorem}

The proof of {\it Theorem \ref{them-2}} is obvious. It is possible for P2P networks to achieve a specific probability $ p $ based on the communication key establishment scheme, for example, 0.999. 
Similarly, the sizes of the key pool and the key ring can be adjusted according to actual demands to achieve both higher security and the inter-client higher connectivity rates according to {\it Lemma \ref{lem-2}}. Thus, the whole data transmission process is well protected.

\section{experimental design}

In this section, we conduct the simulation experiments for the proposed PPT protocol in a spam classification scenario. All simulations are implemented on the same computing environment(Linux Ubuntu 16.04, Intel i7-6950X CPU. 62 GB RAM, and 3.6TB SSD) with Tensorflow, Keras and, PyCryptodome.

In our experiments, There are 200 potential clients in P2P networks and only half of them are target clients. The communication keys are established for all potential clients.
We first train a central original model $ M^{0} $ and distribute $ M^{0} $ to all target clients. Then, the 100 target clients collaboratively train a global model based on the proposed PPT protocol. We also design a dropout simulation to evaluate the robustness.
In the remainder of this section, we give the details of our experiments.

\noindent{\bfseries Database } 

The database used in our experiments consists of two different parts. One part is  
Trec06p which contains 37822 English emails from the real world in 2006. There are 12910 hams and 24912 spams in the main corpus with messages. The other part is Trec07 which is also a real-world English email database consisting of 25220 hams and 50199 spams.

\subsection{Simulation of model training}

Firstly, we train a central original model using Convolutional Neural Network (CNN) \cite{CNN15} for spam classification. In the beginning, we construct a word vocabulary for Trec06p. For each sample, we generated a corresponding word embedding matrix. Then, we divide the Trec06p database into a training set and a testing set in a 3:1 scale uniformly. The training samples are sent to a CNN consisting of two convolution layers, two pooling layers, and three fully connected layers. We set the loss function as the {\em cross-entropy error} and the active function as the {\em sigmoid}. We use SGD for gradient descent, where the learning rate is 0.1. And the central original model is saved as $ M^{0} $. The training process is given in Algorithm 1.

\begin{algorithm}  
	\caption{Original Model Training}  
	\begin{algorithmic}[1] 
		\Require $D$ is the Trec06p dataset; $Z$ is cross-validation times; $CNN$ contains two convolution layers, two pooling layers, and three fully connected layers; $SGD$ is the stochastic gradient descent algorithm.			  
		\Ensure the original model $M^{0}$; $trainSet$ and $testSet$ are the training set and testing set; evaluation result $res$.	
		
		\State $\left\{ Matrix \right\} \gets $ (generate word embedding matrix for every sample in $ D $);	
		
		\State $(trainSet, testSet) \gets $ split $\left\{ Matrix \right\}  $;
		
		\State $S_{i} \gets $ (split $trainSet$ in equal parts of $Z$);
		
		\For {each fold $i=1,2,...,Z$}:
		
		\State $\left\{vSet, tSet\right\} \gets \left\{S_{i}, S-S_{i}\right\} $;
		
		\For {each epoch}:
		
		\State $M_{i}^{0} \gets $ modelFit($CNN,SGD, tSet$);
		
		\State $r_{i} \gets $ modelEvaluate($m_{t}, vSet$);
		
		\EndFor
		
		\EndFor
		
		\State $M^{0} \gets $ averageModel($\left\{ (M_{i}^{0}, r_{i}) | i=1,2,...,Z \right\}$);
		
		\State $res \gets $ modelEvaluate($M^{0}, testSet$);

	\end{algorithmic}  
\end{algorithm} 

Secondly, we survey the e-mail amounts of 100 Gmail users. Our user study involves 100 participants, including 59 males and 41 females whose ages range from 16 to 85. All participants come from different regions and countries, including different skin tones. The participants are recruited using the questionnaire www.wjx.cn by WeChat.

\begin{algorithm}  
	\caption{Local Training}  
	\begin{algorithmic}[1] 
		\Require  $M^{0}$ is the central original model; $100$ clients are indexed by $i$; $D_{i}$ is the data set for each client $ u_{i} $. 
		
		\Ensure the local model set $\{m_{i}|i=0,1,...,99 \}$.
		
		\For {each client $ u_{i} $}:
		
		\State Initialize the central original model $ M^{0} $ on local;
		
		\State $trainSet_{i} \gets $ ($ D_{i} $);
		
		\State $m_{i} \gets $ LocalmodelFit($M^{0},CNN,SGD,trainSet_{i}$);	
		
		\EndFor 
		
	\end{algorithmic}  
\end{algorithm}

We split the Trec07 database into two parts. A portion is assigned to 100 clients as training sets to train their local models. Every client's spam and ham e-mail amounts are assigned strictly according to the results of our user study. Another portion is used as a testing set to evaluate the proposed PPT protocol. 

Thirdly, we simulate the local training processes. 100 target clients re-train $ M^{0} $ using their local training sets respectively. We give the first round local training algorithm flow in Algorithm 2.

\subsection{Simulation of data aggregation}

We simulate the interactions of the proposed PPT protocol. The dropout situation is considered, too. We take the assumption that a fixed number of target clients drop out to reveal the robustness against dropout. 

Firstly, we randomly generate 200 potential clients and their connections to simulate a P2P network topology. Next, we construct communication keys for the 200 potential participants. According to the E-G scheme in the proposed PPT protocol, we set the size of the key pool as 2000 and the key ring size 20.  
We illustrate the connections among potential participants based on the communication keys as the gray lines in Fig. 4. We also generate 100 pairs of secret and public keys $ \{(SK_{i},PK_{i})\vert i=0,1,...,99\} $ for signature.

\begin{figure}[htbp] 
	\centering 
	\includegraphics[scale=0.1]{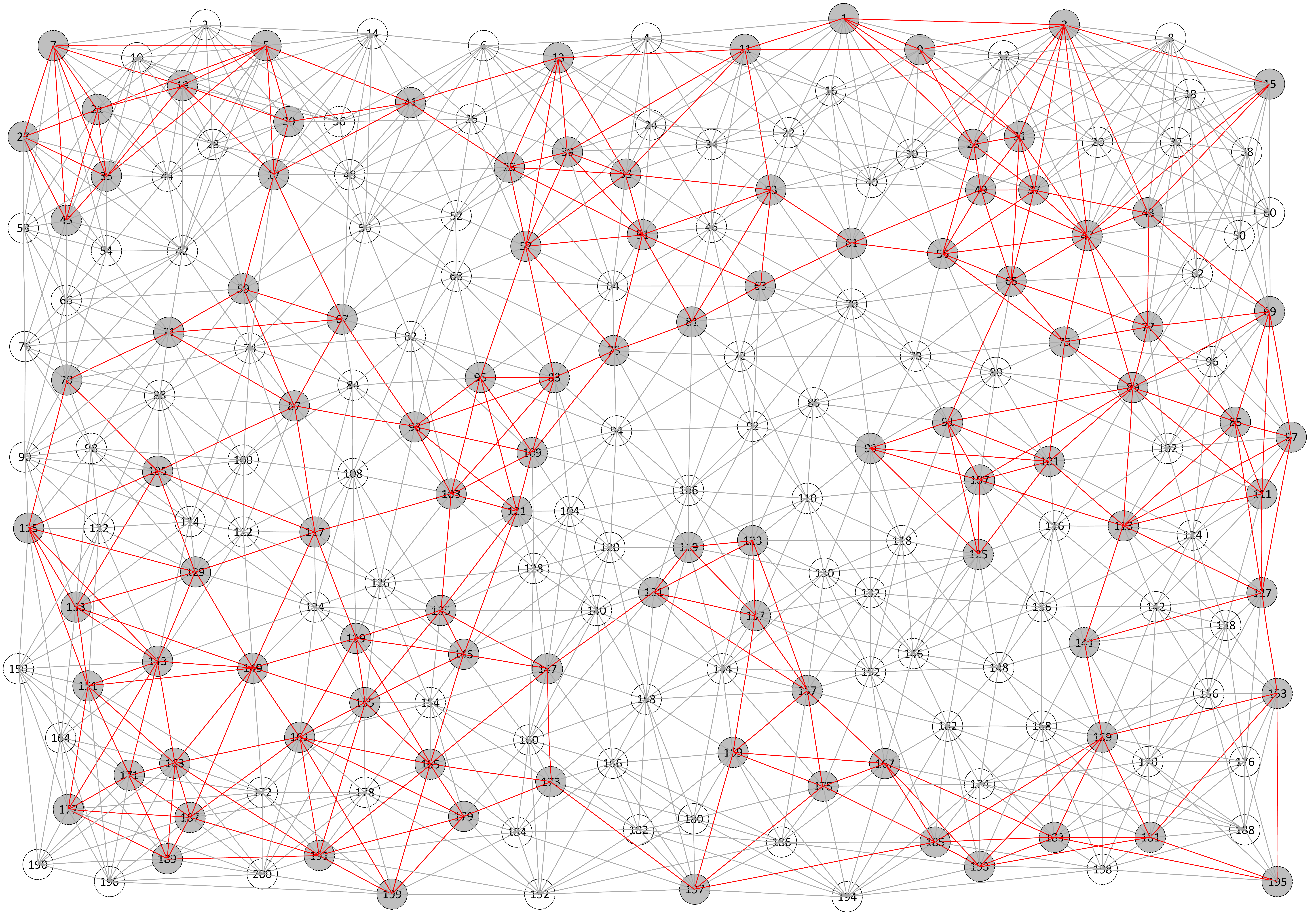} 
	\caption{Clients in a P2P network.} 
\end{figure}

Secondly, we randomly select 100 target clients $ \{u_{i}, i=0,1,...,99\} $, which are shown as the gray node in Fig. 4. Then, we mark one of them as the leader client $u_{0}$. A data transmission route for the 100 target clients according to the {\em depth-first search algorithm} is confirmed subsequently. We illustrate the transmission route as the red line in Fig. 4. 

A noise $s$ is generated to disturb the data of the leader client. In the experiments, we take the AES \cite{AES95} as the encryption algorithm, where the keys are the communication keys. And the signature scheme is designed to use the ElGamal-based signature algorithm \cite{ElGamal85}.

Finally, we aggregate the encoding model update parameters according to the designed transmission route and encryption algorithm and complete the first round of global model training:

\begin{eqnarray}
M^{1}=\frac{\sum_{i=0}^{99} \omega_{i}x_{i}}{\sum_{i=0}^{99} \omega_{i}}+M^{0}.
\end{eqnarray}
Through several rounds of global model training, the global model will converge. 

Afterward, we design a series of dropout clients as the contrast experiments to demonstrate the robustness of the proposed PPT protocol. Besides, we evaluate the efficiency of our privacy-preserving method based on noise addition compared to secret sharing which is deployed in Google's Secure Aggregation protocol \cite{FL17aggregation}. 

\section{Experimental result and Performance}
In this section, we present the experimental results and evaluate the performance of the proposed PPT protocol in terms of correctness, robustness, and efficiency.
The final updated global model is accurate in classifying the samples both in the trec07 testing set and the trec06 testing set. 

After 14 epochs of training in the trec06p training set, the central original model achieves the accuracy of 99.99\% and 99.88\% classifying the samples in the trec06p validation set and testing set, respectively. When classifying the samples in the trec07 testing set, the accuracy is only 76.10\%. In the no dropout experiment settings, the global model converges at the 14th round, which is shown in Fig. 5(a).
And the accuracy of the final global model is 99.12\% in the trec07 testing set, which can be seen in Fig. 5(b). Meanwhile, the final updated global model performs well in the trec06p testing set, whose accuracy is 90.95\%. In Fig. 5., we compare the performances of the central original model and the final global model by illustrating the receiver operating characteristic curve (ROC) and the area under the curve (AUC). The final global model performs better than the central original model in the trec07 and trec06p testing set.

\begin{figure}[htbp]
	\centering
	\subfigure[Loss value in Trec07 testing set.]{
		\begin{minipage}[t]{0.5\linewidth}
			\centering
			\includegraphics[width=1.6in]{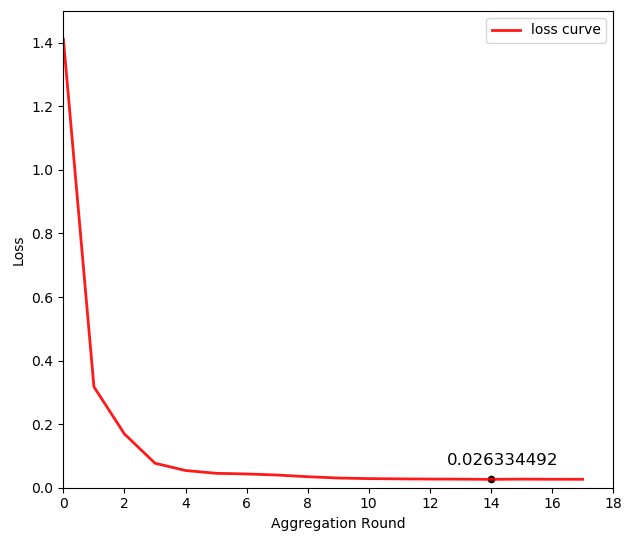}
		\end{minipage}%
	}%
	\subfigure[Accuracy in Trec07 testing set.]{
		\begin{minipage}[t]{0.5\linewidth}
			\centering
			\includegraphics[width=1.6in]{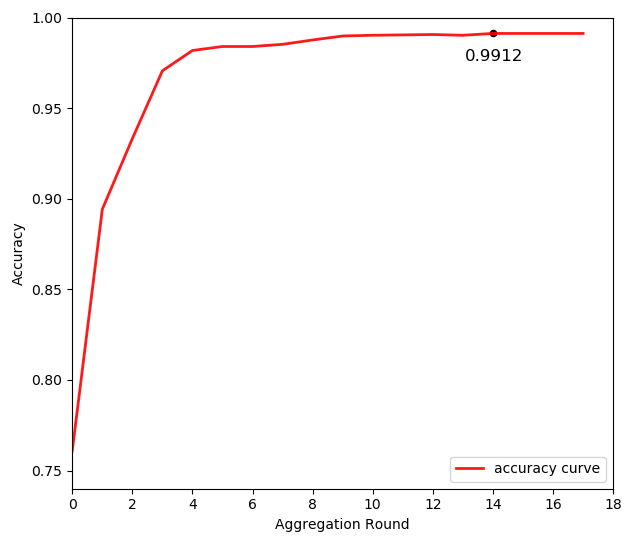}
		\end{minipage}%
	}%
	\centering
	\caption{Performance of the final global model in the Trec07 testing set.}
\end{figure}

\begin{figure}[htbp]
	\centering
	\subfigure[ROC in Trec06p testing set]{
		\begin{minipage}[t]{0.5\linewidth}
			\centering
			\includegraphics[width=1.6in]{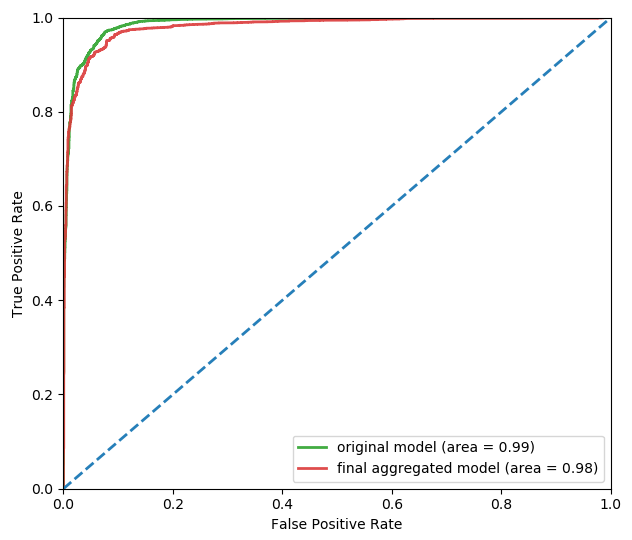}
		\end{minipage}%
	}%
	\subfigure[ROC in Trec07 testing set]{
		\begin{minipage}[t]{0.5\linewidth}
			\centering
			\includegraphics[width=1.6in]{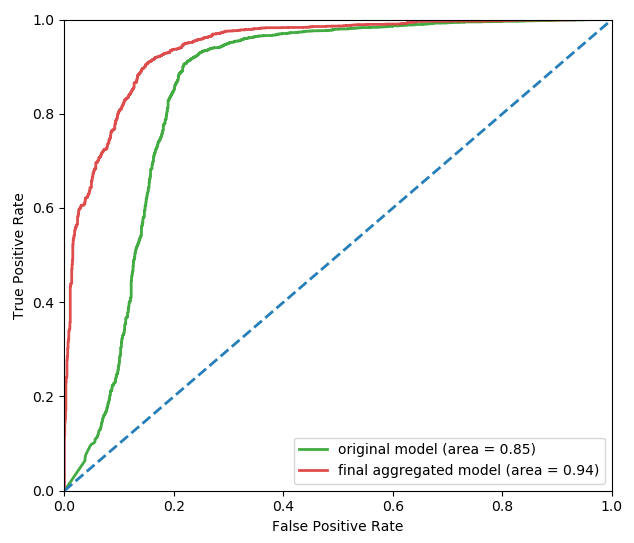}
		\end{minipage}%
	}%
	\centering
	\caption{ROC of the original model and final global model.}
\end{figure}

While randomly choosing dropout clients, the updated models still have good performances. Fig.6. shows the accuracies in the trec07 testing set when the amounts of dropout clients are 1, 5, 10, and 15 respectively. 
Despite dropout, the final global model still achieves an accuracy of 91.02\% at least. 
What's more, the communication key and global model aggregation form support the dropout situation. 
The averaging amounts of one-hop transmission in one aggregation round are shown in TABLE II. 

\begin{figure}[htbp] 
	\centering 
	\includegraphics[scale=0.35]{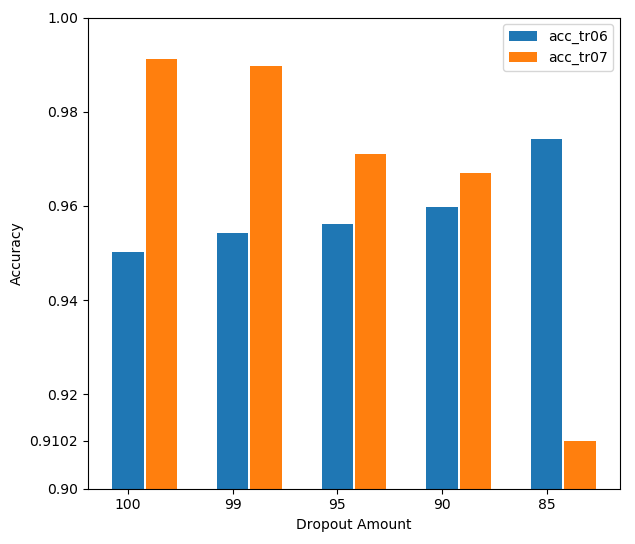} 
	\caption{Accuracy when meeting dropout situation.} 
\end{figure}

\begin{table}[htbp] 
	\centering  
	\caption{Averaging Data transmission amount} 
	\begin{threeparttable}  
		\begin{tabular}{|p{2.5cm}<{\centering}|p{0.5cm}<{\centering}|p{0.5cm}<{\centering}|p{0.5cm}<{\centering}|p{0.5cm}<{\centering}|p{0.5cm}<{\centering}|} 

			\hline  
			Client amount & 100 & 99 & 95 & 90 & 85 \\ 
			\hline 
			Transmission amount & 190 & 185 & 183 & 171 & 163  \\
			
			\hline

		\end{tabular}

	\end{threeparttable}  
\end{table}  

The computational performance for one client is exhibited in TABLE III. 
Compared to the secret sharing deployed in Google's Secure Aggregation protocol, The PPT protocol requests less computational resources and offers more efficiency for clients. 

\begin{table}[h] 
	\centering  
	\caption{performance} 
	\begin{threeparttable}  
		\begin{tabular}{m{1cm}<{\centering}|m{3.4cm}<{\centering}|m{3cm}<{\centering}} 
			
			\hline  
			\hline  
			\makecell[c]{Index} & \makecell[c]{Operations\tnote{1}} & \makecell[c]{Time(ms/1000byte)} \\ 
			\hline 
			\makecell[c]{1} & \makecell[c]{Secret sharing} & \makecell[c]{23.3164}\\  
			\hline  
			\makecell[c]{2} & \makecell[c]{Secret reconstruction} & \makecell[c]{9.8632}\\  
			\hline  
			\makecell[c]{3} & \makecell[c]{Noise generation} & \makecell[c]{1.2548}\\  
			\hline  
			\makecell[c]{4} & \makecell[c]{Noise addition} & \makecell[c]{0.1422}\\  
			\hline  
			\makecell[c]{5} & \makecell[c]{Noise subtraction} & \makecell[c]{0.1506}\\  
			\hline  
			\makecell[c]{6} & \makecell[c]{Encryption(AES-128bit) } & \makecell[c]{170.8248}\\
			\hline  
			\makecell[c]{7} & \makecell[c]{Decryption(AES-128bit) } & \makecell[c]{0.0282}\\
			\hline  
			\makecell[c]{8} & \makecell[c]{Signature(Elgamal-2048bit)} & \makecell[c]{0.0003}\\
			\hline
			\makecell[c]{9} & \makecell[c]{Verification(Elgamal-2048bit)} & \makecell[c]{0.0071}\\  
			\hline  
			\hline 
			
		\end{tabular}
		
		\begin{tablenotes}
			\footnotesize
			\item[1] We execute a series of operations on a 114MB file of local model updating parameters in the form of plaintext. And the encrypted result is 440.1MB.   
		\end{tablenotes}
	\end{threeparttable}  
\end{table}  

\section{Limitation and Future work}
In this section, we briefly discuss the limitation of the proposed PPT protocol. We also have a vision of the future works.

PPT is designed for the {\em FL in P2P networks} where clients are usually not directly connected with a central server. Despite inheriting the data transmission speed from P2P networks, the one-hop transmission form is obviously less efficient. When the amount of participants is larger enough, the communication efficiency of PPT is not so high. What's more, the communication key establishment scheme is more proper for a static network. In some dynamic P2P networks, especially Internet of Vehicles (IoV) and Ad Hoc Networks scenarios, clients are always mobile. Thus, the future research direction point to the efficiency of {\em federated learning in large-scale and dynamic P2P networks}. 

The supervision and report mechanism is based on a pure strategy game, i.e., adversaries' strategies are equiprobable. Nevertheless, adversaries' strategies are complex in practice, especially the probability of different strategies. Improving the supervision and reporting mechanism based on Game Theory more in line with the real scenario will be our next research.

Besides, ecological validity is a challenge to our user study. Our study mainly recruits students in the university. These participants are usually more active in using e-mail applications. Thus the performance evaluation may vary with other populations. In future works, we will conduct a large-scale user study involving more participants to perform a more intensive evaluation of our protocol.

\section{Conclusion}

The status lacking a trusted central server necessitates the development of {\em federated learning in P2P networks}. This research provides an instance of such a context, along with guarantees on both its communication efficiency and privacy. While there have been prior works on federated learning and collaboratively training, our work is the first to provide a privacy-preserving global model training protocol in the none central aggregator setting from the secure and efficient transmission perspective. Besides, our training protocol is dropout-robust, which is of practical significance. And our experiments conducted on two real-world datasets suggest that our training protocol is practical.


%

\appendices

\ifCLASSOPTIONcaptionsoff
  \newpage
\fi

%






\end{document}